\documentclass[sigconf]{acmart}

\pdfoutput=1

\usepackage[ruled,vlined,linesnumbered]{algorithm2e}

\usepackage{tikz} 
\usetikzlibrary{shapes, arrows, matrix, positioning, patterns, decorations.pathreplacing}

\usepackage{pgfplots}
\pgfplotsset{
    compat=newest,
    table/col sep=comma,
    table/x=x,
    table/y=y
}

\usepackage{makecell}
\usepackage{subcaption}

\usepackage{paralist}

\usepackage{ifthen}

\theoremstyle{plain}
\newtheorem{theorem}{Theorem}[section] 
\newtheorem{lemma}[theorem]{Lemma} 
\newtheorem{proposition}[theorem]{Proposition} 
\newtheorem{observation}[theorem]{Observation}

\theoremstyle{definition}
\newtheorem{definition}[theorem]{Definition}

\newcommand{\Sum}[1]{\ensuremath{\mathsf{Sum}(#1)}}

\newcommand{\Ran}[1]{\ensuremath{\mathit{range}(#1)}}

\newcommand{\Dfn}[1]{\textbf{\emph{#1}}}

\newcommand{\PSG}{\ensuremath{\mathit{PSG}}\xspace}

\newcommand{\PSRM}{\ensuremath{\mathit{PSRM}}\xspace}

\newcommand{\Entities}{\ensuremath{\mathit{E}}\xspace}
\newcommand{\DTP}{\ensuremath{\mathit{DTP}}\xspace}
\newcommand{\TBL}{\ensuremath{\mathit{tbl}}\xspace}
\newcommand{\PSM}{\ensuremath{\mathit{PSM}}\xspace}

\newcommand{\eqsub}{\ensuremath{\equiv_{\mathit{row}}}\xspace}
\newcommand{\eqobj}{\ensuremath{\equiv_{\mathit{col}}}\xspace}

\newcommand{\TARGETDOC}{0}

\begin{document}
\title{Mining Domain-Based Policies}


\ifthenelse{\TARGETDOC = 0}{
\author{Si Zhang}
\affiliation{
    \institution{Department of Computer Science\\University of Calgary}
    \city{Calgary}
    \country{Canada}
}
\email{si.zhang2@ucalgary.ca}

\author{Philip W. L. Fong}
\affiliation{
    \institution{Department of Computer Science\\University of Calgary}
    \city{Calgary}
    \country{Canada}
}
\email{pwlfong@ucalgary.ca}
}{
  \author{Anonymized}
}

\begin{abstract}
  Protection domains are one of the most enduring concepts in Access
  Control. Entities with identical access control characteristics are
  grouped under the same protection domain, and domain-based policies
  assign access privileges to the protection domain as a whole. With
  the advent of the Internet of Things (IoT), devices play the roles
  of both subjects and objects. Domain-based policies are particularly
  suited to support this symmetry of roles.

  This paper studies the mining of domain-based policies from
  incomplete access logs. We began by building a theory of
  domain-based policies, resulting in a polynomial-time algorithm that
  constructs the optimal domain-based policy out of a given access
  control matrix. We then showed that the problem of domain-based
  policy mining (DBPM) and the related problem of mining policies for
  domain and type enforcement (DTEPM) are both NP-complete. Next, we
  looked at the practical problem of using a MaxSAT solver to solve
  DBPM. We devised sophisticated encodings for this purpose, and
  empirically evaluated their relative performance. This paper thus
  lays the groundwork for future study of DBPM.
\end{abstract}



\maketitle

\section{Introduction}

Protection domains are one of the most enduring concepts in the study
of Access Control. Entities with identical access control
characteristics are assigned to the same protection domain, and access
privileges are granted to the protection domain as a whole.  We use
the term \Dfn{domain-based policies} to refer to access control
policies that are formulated in terms of protection domains.  The
use of protection domains can be found in programming language
environments (e.g., Java), databases \cite{Baldwin1990}, and operating
systems. For example, at the core of the Android operating system is
SELinux, which implements \Dfn{domain and type enforcement}
\cite{Badger-etal:1995}.
  
With the advent of the Internet of Things (IoT), devices play the
roles of both subjects and objects.  In other words, a device may
request access to other devices but also act as a resource to be
accessed by others.  An access occurs when a requestor device sends a
message to a resource-bearing device.  Domain-based policies is
particularly suited to support this symmetry of roles.  Recognizing
this advantage, Fuentes Carranza and Fong \cite{Carranza2019} extended
the MQTT broker to provide fine-grained access control for IoT
devices.  More specifically, their extension of the event-based system
architecture \cite{Fiege2002C,Fiege2002J} involves two protection
mechanisms.  First, their extension enforces \Dfn{brokering policies},
so that one can control which network link the broker forward a
message received from another network link.  Second, \Dfn{execution
  monitors} \cite{Schneider2000,Ligatti2009} can be interposed on
network links, so that messages can be suppressed or transformed.
When combined, these mechanisms can enforce a form of domain-based
policies over devices.

This work studies the mining of domain-based policies from access
logs.  We take into consideration that in IoT systems, each entity
(device) plays the dual role of requestor and resource. Access control
is about controlling who can send which kind of message to whom.  The
authoring of access control policies, however, is a highly demanding
task if it is done from scratch.  It is therefore desirable to
automatically mine a domain-based policy from existing access
logs. When the access logs are complete, they allow us to reconstruct
the entire access control matrix. Policy mining then corresponds to
the construction of a domain-based policy from a fully specified
access control matrix.  When the access logs are incomplete, the
corresponding access control matrix will have missing entries. Such a
missing entry can be interpreted as ``don't care.''  This, in theory,
gives the policy mining algorithm more freedom to maneuver, in the
sense that the algorithm may interpret each missing entry as either a
grant or a deny.  Every possible instantiation of the missing entries
leads to the construction of a different domain-based policy.  We
desire an instantiation that results in the smallest number of
protection domains.

We claim the following contributions:
\begin{asparaenum}
\item In \S \ref{sec:algorithm}, we developed a theory of domain-based
  policies, in which we characterized domain-based policies in terms
  of digraph strong homomorphism and an efficiently checkable
  equivalence relation.  Armed with this understanding, we devised a
  polynomial-time algorithm for constructing a domain-based policy
  from an access control matrix.  This is essentially an algorithm for
  mining domain-based policies from \emph{complete} access logs.
\item In \S \ref{sec:formalization}, we showed that the problem of
  mining domain-based policies from \emph{incomplete} access logs
  (DBPM) is NP-complete.
\item Generalizing domain-based policies, domain and type enforcement
  \cite{Badger-etal:1995} classifies an entity in two ways: (a) its
  subject characteristics (domain) and (b) its object characteristics
  (type). We showed, in \S \ref{sec:DTE}, that while this dual
  classification could lead to a small number of domain/type labels,
  the policy mining problem remains NP-complete.
\item In \S \ref{sec:encoding}, we addressed the practical challenge
  of solving DBPM. To this end, we devised multiple MaxSAT encodings
  for DBPM. These encodings involve sophisticated optimization
  techniques (e.g., symmetry breaking) that are not attempted in the
  policy analysis and policy mining literature.
\item In \S \ref{sec:evaluation}, we empirically evaluate the
  aforementioned encodings, and demonstrate the utility of the
  optimization techniques.
\end{asparaenum}

\section{A Theory of Domain-Based Policies}
\label{sec:algorithm}

Suppose we have a collection of access logs from the past, telling us
which devices may (or may not) send messages of certain topics to
which other devices.  The access logs correspond to entries in the
access control matrix.  The construction of a domain-based policy from
these access logs is what we call \Dfn{domain-based policy mining}.
In practice, the logs do not fully specify the access control matrix:
i.e., some of the entries in the access control matrix are left
unspecified. The general problem of mining a domain-based policy from
a partially specified access control matrix is going to be studied in
\S \ref{sec:formalization}. In this section, we explore a special case
of that general problem, the case when the access logs fully specify
an access control matrix. While this special case is not typical in
practice, this study gives us the opportunity to develop a theory of
domain-based policies. Such a theoretical understanding will be
exploited in \S \ref{sec:formalization} for studying the general
problem.

Proofs of propositions and theorems in this section can be found in
Appendix \ref{app:proofs}.

\subsection{From Extension to Abstraction}

In the literature of Access Control, there are several styles of
policy specification. A first style is the \Dfn{extensional}
specification of a policy, in which all permitted access requests are
exhaustively enumerated.  The classical example of this style of
specification is the access control matrix \cite{Lampson1974}. A
second style of specification is to create \Dfn{access control
  abstractions} to facilitate policy administration.\footnote{Modern
  access control paradigms may also exploit \Dfn{intensional} policy
  specifications, in which a policy expresses the condition of access
  in a declarative manner. Examples include ABAC \cite{Hu2015} and
  ReBAC \cite{Fong2011}.}  A classical example is the creation of
roles as an intermediary concept \cite{Sandhu1996}.  Other access
control abstractions include protection domains, user groups, types
\cite{Badger-etal:1995}, demarcations \cite{Kuijper2014}, and
categories \cite{Barker:2009}.  The general idea is that entities with
identical access control characteristics are grouped into higher-level
concepts, so that the same access restrictions are applied uniformly
to the entire group. This practice makes policy administration more
tractable, and decouples the management of concept membership from the
assignment of permissions.  Mining domain-based policies is
essentially the inference of the second type of policy specification
(abstraction-based) from the first type of policy specification
(extensional). In the following, we formalize extensional policy
specification and abstraction-based policy specification.

Following Graham-Denning \cite{Graham1972} and Lipton-Snyder
\cite{Lipton1977, Snyder1981}, we represent access
control matrices by edge-labelled directed graphs.
\begin{definition} \label{def:eldigraph} Fix a finite set $\Sigma$ of
  \Dfn{access rights}.  An \Dfn{edge-labelled directed graph} (or
  \Dfn{digraph} for brevity) $G$ is a pair $(V, E)$, where
  $E \subseteq V \times \Sigma \times V$. We write $V(G)$ and $E(G)$
  for the sets $V$ and $E$.
\end{definition}

Intuitively, a digraph specifies an access control policy in an
extensional manner (i.e., exhaustive enumeration of all permitted
access requests). In particular, the vertices are either subjects or
objects. If there is an edge $(u, a, v)$, then $u$ is permitted to
exercise access right $a$ against $v$.  Otherwise, the absence of an
edge specifies a prohibition.  The edge label $a$ can be interpreted
in many ways: access right, access mode, access event, method
invocation, message topic, etc.  For instance, in Fuentes Carranza and
Fong \cite{Carranza2019}, an edge $(u,a,v)$ could be interpreted as
the access request ``device $u$ sends a message of topic $a$ to device
$v$.''  In this sense, a digraph carries the same information as an
access control matrix.

Standard graph-theoretic concepts, such as subgraphs, isomorphism,
etc, can be defined accordingly. We only highlight the notion of
induced subgraphs, for it plays an important role in the sequel: Given
$U \subseteq V(G)$, we write $G[U]$ for the \Dfn{subgraph of $G$
  induced by $U$}: i.e., the digraph
$(U, E(G) \cap (U\times\Sigma\times U))$.

The mining of domain-based policies can be seen as the process of
moving from the exhaustive enumeration of accessibility to a more
succinct summary of the same policy. In particular, summarizing $G$
involves several steps:
\begin{enumerate}
\item Create a digraph $H$.  Each vertex of $H$ represents a
  protection domain. An edge $(u, a, v) \in E(H)$ indicates that an
  entity assigned to domain $u$ is permitted to perform access $a$ to
  any entity assigned to domain $v$.
\item Create a \Dfn{protection domain assignment}
  $\pi : V(G) \rightarrow V(H)$ that assigns each entity (i.e., a
  vertex from $G$) to a protection domain (i.e., a vertex from $H$).
\item When an access request $(u, a, v)$ is issued, the protection
  mechanism grants access if the edge $(\pi(u), a, \pi(v))$ is in
  digraph $H$, and denies access otherwise.
\end{enumerate}
Applying this scheme to the work of Fuentes Carranza and Fong
\cite{Carranza2019}, vertices of $H$ are the brokers, edges of $H$
specify accessibility, and $\pi$ assigns devices to brokers.
\begin{definition}[Domain-Based Policy]
  Given a digraph $G$, a \Dfn{do\-main-based policy (for $G$)} is a pair
  $(H, \pi)$, where $H$ is a digraph and $\pi: V(G) \rightarrow V(H)$
  is a protection domain assignment.
\end{definition}
While $G$ can be very large, $H$ is expected to be of manageable
size. This arrangement has multiple advantages: (i) the management of
$\pi$ (i.e., domain membership) can be decoupled from the
administration of $H$ (i.e., specification of accessibility), and (ii)
the administration of $H$ is more tractable than the direct
administration of $G$.

A potential problem with the above scheme is that $H$ may not properly
summarize the accessibility expressed in $G$.  It may grant an access
when the latter should have been denied, or deny accesses that should
have been permitted.
\begin{definition}[Enforcement]
  Domain-based policy $(H, \pi)$ \Dfn{enforces} digraph $G$ whenever
  the following holds: for $u,v \in V(G)$ and $a \in \Sigma$,
  $(u, a, v) \in E(G)$ iff $(\pi(u), a, \pi(v)) \in E(H)$.
\end{definition}
For $H$ to be a ``correct'' summary of $G$, the assignment mapping
$\pi$ needs to be a strong homomorphism.
\begin{definition}
  Given digraphs $G$ and $H$, a \Dfn{strong homomorphism} from $G$ to
  $H$ is a function $\pi : V(G) \rightarrow V(H)$ such that
  $(u, a, v) \in E(G)$ iff $(\pi(u), a, \pi(v)) \in E(H)$. $G$ is
  \Dfn{strongly homomorphic} to $H$ iff there is a strong homomorphism
  from $G$ to $H$.
\end{definition}
The next result follows immediately from the definition above.
\begin{proposition} \label{Prop:Fault-Empty} A domain-based policy
  $(H, \pi)$ enforces a digraph $G$ iff $\pi$ is a strong homomorphism
  from $G$ to $H$.
\end{proposition}
Essentially, mining a domain-based policy for $G$ involves finding (i)
a digraph $H$ and (ii) a strong homomorphism $\pi$ from $G$ to $H$.

\subsection{Most Succinct Summary}
\label{sec:most-succinct}

Every digraph $G$ is strongly homomorphic to itself (via an
isomorphism). Such a ``summary'' is not succinct at all. What we
desire is for the summary to compress information as much as
possible. The best summary is a digraph that cannot be further
summarized.
\begin{definition} \label{def:summary} Finite digraph $H$ is a
  \Dfn{summary} of finite digraph $G$ iff (a) $G$ is strongly
  homomorphic to $H$, and (b) $G$ is not strongly homomorphic to any
  proper subgraph of $H$.
\end{definition}

The definition above has three important consequences. The first is that
the strong homomorphism from a digraph to its summary is always a
surjection, meaning that a summary contains no redundant vertices.
\begin{proposition} \label{prop:surjective}
  If $H$ is a summary of $G$ via the strong homomorphism $\pi$, then
  $\pi$ is surjective.
\end{proposition}

The second consequence of Definition \ref{def:summary} is that a
summary is always ``minimal,'' in the sense that the summary cannot be
further summarized.  
\begin{definition}[Irreducible Digraph] \label{def:irreducible} A
  digraph $G = (V, E)$ is \Dfn{irreducible} iff every summary $H$ of
  $G$ is isomorphic to $G$.
\end{definition}
It can be shown that summaries are irreducible.
\begin{proposition} \label{prop:irreducible}
  Suppose there is a surjective strong homomorphism
  from digraph $G$ to digraph $H$. Then $H$ is a
  summary of $G$ iff $H$ is irreducible.
\end{proposition}

The third consequence of Definition \ref{def:summary} is that
summaries are unique.
\begin{proposition} \label{prop:unique}
  The summary of a digraph $G$ is unique up to isomorphism.
\end{proposition}

\subsection{Constructing Digraph Summary}

In the rest of this section, we examine how we can compute a
domain-based policy $(H, \pi)$ for $G$, such that $H$ is a summary of
$G$. This corresponds to the construction of a domain-based policy
from complete access logs ($G$).  While Definition \ref{def:summary} is a
``global'' characterization of a digraph summary, we do not yet have a
way to construct this summary algorithmically.  An important
contribution of this section is the proposal of a ``local'' means for
computing a digraph summary. The algorithmic construction is based on
an equivalence relation defined over the vertex set of $G$, a relation
that can be tested efficiently.

\begin{definition}[Indistinguishable Vertices] 
    \label{def:indistinct}
    Two vertices $u$, $v$ in a digraph $G$ are \Dfn{indistinguishable}
    iff both of the following conditions hold for every $a \in \Sigma$:
    \begin{enumerate}
    \item Either the following four edges all belong to $E(G)$ or they do
      not: $(u, a, u)$, $(u, a, v)$, $(v, a, u)$, $(v, a, v)$.
    \item For every $x \in V(G) \setminus \{ u, v \}$,
      \begin{enumerate}
      \item $(u, a, x) \in E(G)$ iff $(v, a, x) \in E(G)$, and
      \item $(x, a, u) \in E(G)$ iff $(x, a, v) \in E(G)$.
      \end{enumerate}
    \end{enumerate}
    We write $u \equiv_G v$ (or simply $u \equiv v$ when there is no
    ambiguity) to assert that $u$ and $v$ are indistinguishable.
\end{definition}
In other words, two vertices are indistinguishable in a digraph when
(i) edges with a label $a$ either form a complete or empty subgraph
between the two vertices; and (ii) their adjacencies (and
non-adjacencies) with other vertices are the same.  A special case is
$u \equiv v$ whenever $u = v$.

\begin{proposition}
  \label{prop:equiv_relation}
  The relation $\equiv_G$ is an equivalence relation.
\end{proposition}

When $G$ is finite, checking if two vertices are equivalent takes time
linear to $|V(G)|$, a computational advantage we will exploit later in
devising algorithms.  Not only that, the equivalence relation also
provides us with a way to construct the summary of $G$.
\begin{definition}
  Suppose $G$ is a digraph, and $\equiv$ is $\equiv_G$. Recall that,
  given $v \in V(G)$, $[v]_\equiv$ is the equivalence class containing
  $v$.  Then \Sum{G} is defined to be the digraph $(V, E)$ such that
  $V = \{\, [v]_\equiv \mid v \in V(G) \,\}$ and
  $E = \{\, ([u]_\equiv, a, [v]_\equiv) \mid (u,a,v) \in E(G) \,\}$.
\end{definition}
We are now ready to state the main result of this section: a summary
always exists for any digraph, and it is isomorphic to \Sum{G}.  This
theorem provides a programmatic way to construct a summary.
\begin{theorem} \label{thm:offline-correctness}
  \Sum{G} is the summary of $G$.
\end{theorem}

\begin{algorithm}[t] \label{alg:ftu}
  \caption{\textsc{Summarize}($G$)\label{algo:offline}}
  \KwIn{a digraph $G$} 
  \KwOut{
    a summary $H$ of $G$ and the corresponding strong
    homomorphism $\pi$
  }

  \lForEach{$u \in V(G)$}{\textsc{Make-Set}($u$)\label{line:equiv-begin}}
  \ForEach{unorder pair $\{u,v\} \subseteq V(G)$\label{line:equiv-for-pair}}{
    \If{$\textsc{Find-Set}(u) \neq \textsc{Find-Set}(v)$ and $u \equiv_G v$\label{line:test-equiv}}{%
      \textsc{Union}($u$, $v$)\label{line:equiv-end}\;%
    }%
  }
  let $U = \{\, \textsc{Find-Set}(u) \,\mid\, u \in
  V(G)\,\}$\label{line:build-summary}\;
  define $\pi$ such that $\pi(u) =
  \textsc{Find-Set}(u)$\label{line:build-homo}\;
  \Return $(G[U], \pi)$\label{line:return-summary}\;
\end{algorithm}

Theorem \ref{thm:offline-correctness} presents a tractable way to
construct a summary of $G$. The idea, captured in Algorithm
\ref{algo:offline}, is to compute the equivalence classes induced by
$\equiv_G$ and then return an isomorph of \Sum{G}.

Lines \ref{line:equiv-begin}--\ref{line:equiv-end} compute the
equivalence classes of vertices in $G$. We employ the disjoint set
data structure \cite[Ch.~21]{Cormen2009}, so that every disjoint set
models an equivalence class.  We start by putting every vertex in its
own disjoint set (line \ref{line:equiv-begin}). Then we go through
every unordered pair of vertices (line \ref{line:equiv-for-pair}). If
the two vertices belong to separate disjoint sets but they are
equivalent (line \ref{line:test-equiv}), then their disjoint sets are
merged (line \ref{line:equiv-end}).  When the loop terminates, the
disjoint sets are exactly the equivalence classes.

The summary digraph and the corresponding strong homomorphism are
constructed and returned in lines
\ref{line:build-summary}--\ref{line:return-summary}.  Rather than
returning \Sum{G} directly, we return a subgraph of $G$ that is
isomorphic to \Sum{G}.
Specifically, we select a representative member from each equivalence
class.  Since we are already using the disjoint set data structure, we
designate $\textsc{Find-Set}(u)$ to be the representative of
$[u]_{\equiv_G}$.  We collect all representatives in the set $U$ (line
\ref{line:build-summary}), and eventually return $G[U]$ as the summary
of $G$ (line \ref{line:return-summary}).  Also, line
\ref{line:build-homo} sets up the strong homomorphism $\pi$ to map $u$
to its representative $\textsc{Find-Set}(u)$.

To evaluate the running time of Algorithm \ref{algo:offline}, let
$n = |V(G)|$ and $k = |\Sigma|$, and assume that $G$ is represented as
an adjacency matrix.  The running time of Algorithm \ref{algo:offline}
is dominated by the for-loop in lines
\ref{line:equiv-for-pair}--\ref{line:equiv-end}.  That for-loop
iterates for $n \choose 2$ times.  In each iteration, it performs (a)
a number of disjoint-set operations and (b) an equivalence test
($\equiv_G$). By Definition \ref{def:indistinct}, an equivalence test
can be performed in an adjacency matrix in $O(kn)$ time.  Disjoint-set
operations takes $O(\log n)$ time if \Dfn{union-by-rank} is
implemented \cite[Ch.~21]{Cormen2009}.  The running time of
equivalence tests dominates. Therefore, the running time of Algorithm
\ref{algo:offline} is $O(kn^3)$.

\section{Mining a Domain-Based Policy}
\label{sec:formalization}

The last section demonstrates that the construction of a domain-based
policy from complete access logs can be performed efficiently. In this
section, we study the mining of domain-based policies in a more
general setting, one in which the access logs are incomplete. To that
end, we begin by formalizing incomplete access logs.
\begin{definition}
  Fix a finite set $\Sigma$ of access rights.
  \begin{itemize}
  \item A \Dfn{partially specified digraph} \PSG is a triple
    $(V, E, F)$, where the sets $E$ and $F$ are disjoint subsets of
    $V \times \Sigma\times V$.  We write $V(\PSG)$, $E(\PSG)$,
    $F(\PSG)$, and $D(\PSG)$ for the sets $V$, $E$, $F$, and
    $(V \times \Sigma\times V) \setminus (E\cup F)$ respectively.
    (Intuitively, $E(\PSG)$ contains the \Dfn{edges} of $\PSG$,
    $F(\PSG)$ contains the \Dfn{non-edges} (absence of edges), and
    $D(\PSG)$ contains the ``don't care'' (missing data). The
    intention is that \PSG represents incomplete access logs.)
  \item A digraph $G$ is an \Dfn{instantiation} of \PSG if and only if
    $E(\PSG) \linebreak \subseteq E(G)$ and $F(\PSG) \cap E(G) = \emptyset$.
  \end{itemize}
\end{definition}
Our goal is to construct, out of \PSG, a domain-based policy that
grants the access requests that \PSG grants, and denies the access
requests that \PSG denies.  Due to the incomplete nature of \PSG,
there are access requests that are neither granted nor denied by
\PSG. We want our domain-based policy to generalize its authorization
decisions to cover such access requests.  In the literature of policy
mining \cite{Mitra2016}, the common practice is to adopt the Occam's
Razor principle \cite[Ch.~2]{Kearns1994}, and attempt to minimize the
complexity of the constructed policy.  Following this
practice, we construct a domain-based policy with the smallest number
of domains. This decision form of this optimization problem is stated
below.

\textbf{Domain-Based Policy Mining (DBPM)}
\begin{itemize}
\item \textbf{Instance:} A positive integer $m$ and a
  partially specified digraph \PSG.
\item \textbf{Question:} Is there an instance $G$ of \PSG such that
  $\equiv_G$ induces no more than $m$ equivalence classes?
\end{itemize}

The \Dfn{graph sandwich problem for property~$\Pi$} asks: Given graph
$G_1$ and $G_2$ such that $V(G_1) = V(G_2)$ and
$E(G_1) \subset E(G_2)$, does there exist a graph $G$ such that
$V(G) = V(G_1)$, $E(G_1) \subseteq E(G) \subseteq E(G_2)$, and $G$
satisfies property $\Pi$ \cite{Golumbic1995}?  DBPM can be seen as a
graph sandwich problem, in which $G_1$ is obtained by turning all
triples in $D(\PSG)$ to non-edges, $G_2$ is obtained by treating the
members of $D(\PSG)$ as edges, and the property $\Pi$ is ``no more
than $m$ equivalence classes.''  Some graph sandwich problems are
decidable in polynomial time, while others are intractable.  DBPM
belongs to the latter kind.
  
\begin{theorem}
  DBPM is NP-complete.
\end{theorem}

\begin{proof}
  \emph{DBPM is in NP.} A nondeterministic algorithm determines if
  \PSG has an instantiation with no more than $m$ equivalence classes
  by first guessing an instantiation $G$ of \PSG, then employing the
  \textsc{Summarize} algorithm to compute $m^*$, the number of
  equivalence classes induced by $\equiv_G$, and lastly checking that
  $m^* \leq m$.

  \emph{DBPM is NP-hard.} We present a polynomial-time reduction
  from Graph 3-Colorability \cite[p.~191]{Garey-Johnson:1979} to DBPM.
  We denote the three colors by $1$, $2$, and $3$.

  Given a Graph 3-Colorability instance $H = (V, E)$, which is an
  undirected graph, the reduction generates a DBPM instance
  $(3 \times |V|, \PSG)$, where \PSG is defined as follows:
  \begin{itemize}
  \item $\Sigma = \{ a_v \,\mid\, v \in V \} \cup
    \{ b_e \,\mid\, e \in E \}$.
  \item $V(\PSG)$ contains the following vertices:
    \begin{itemize}
    \item $x_{v, i}$ for every $v \in V$ and $i \in \{1,2,3\}$
    \item $y_v$ for every $v \in V$
    \item $z_{e, i}$ for every $e \in E$ and $i \in \{1,2,3\}$
    \end{itemize}
  \item $E(\PSG)$ contains the following edges:
    \begin{itemize}
    \item $(x_{v,i}, a_v, x_{v, i})$ for every $v \in V$ and
      $i \in \{1,2,3\}$
    \item $(y_v, a_v, y_v)$ for every $v \in V$
    \item $(z_{e,i}, b_e, z_{e, i})$ for every $v \in V$ and
      $i \in \{1,2,3\}$
    \end{itemize}
  \item $F(\PSG)$ contains the non-edges below:
    \begin{itemize}
    \item $(x_{v, i}, a_u, x_{v, i})$ for $u,v \in V$ s.t.~$u \neq v$,
      and for $i \in \{1,2,3\}$
    \item $(x_{v, i}, a_v, x_{v, j})$ for every $v \in V$ and
      $i, j \in \{1,2,3\}$ s.t.~$i\neq j$
    \item $(y_v, a_u, y_v)$ for every $u,v \in V $ s.t.~$u\neq v$
    \item $(y_v, b_e, y_v)$ for every $v \in V$ and $e \in E$ s.t.~
      $v$ is one of the two ends of $e$
    \item $(z_{e, i}, a_u, z_{e, i})$ for every $u \in V$ and
      $e \in E$ s.t.~$u$ is not one of the two ends of $e$
    \item $(z_{e,i}, a_v, x_{v, j})$ and $(x_{v, j}, a_v, z_{e, i})$
      for every $v \in V$, $e \in E$, and $i,j \in \{1,2,3\}$ s.t.~$v$
      is one of the two ends of $e$ but $i \neq j$
    \end{itemize}
  \end{itemize}
  A few observations can be made about and instantiation $G$ of \PSG:
  \begin{enumerate}
  \item $x_{v,i}$ represents the option of coloring vertex $v$ of
    $H$ by color $i$. When $u \neq v$ or $i \neq j$,
    $x_{u,i} \not\equiv_G x_{v,j}$. Thus the vertices $x_{v,i}$ in $G$
    belong to $3 \times |V|$ distinct equivalence classes.
  \item $y_v$ may belong to equivalence class $[x_{u, i}]_{\equiv_G}$
    only when $u = v$.  When $y_v \equiv_G x_{v, i}$, vertex $v$ in
    $H$ is assigned the color $i$.
  \item Suppose $e = uv$. Then $z_{e,i}$ may belong to either
    $[x_{u, i}]_{\equiv_G}$ or $[x_{v,j}]_{\equiv_G}$ (but not other
    equivalence classes of the form $[x_{w, k}]_{\equiv_G}$).
    Yet $z_{e,i} \not\equiv_G y_v$ and $z_{e, i} \not\equiv_G y_u$.
    Therefore, $z_{e,i}$ can only belong to either
    $[x_{u, i}]_{\equiv_G}$ or $[x_{v,j}]_{\equiv_G}$ if
    the same color $i$ is not assigned to the two
    ends $u$ and $v$ of $e$.
  \end{enumerate}
  
  Suppose $\pi : V \rightarrow \{1,2,3\}$ is a 3-coloring of $H$.  We
  can instantiate \PSG to obtain a digraph $G$, such that
  $y_v \equiv_G x_{v, i}$ iff $\pi(v) = i$. Since $\pi$ is a
  3-coloring, the instantiation $G$ can be chosen in such a way that
  for every edge $e = uv \in E$, either $z_{e, i} \equiv_G x_{u,i}$ or
  $z_{e,i} \equiv_G x_{v, i}$.  Thus $\equiv_G$ induces no more than
  $3 \times |V|$ equivalence classes.

  Conversely, suppose \PSG has an instantiation $G$ with no more than
  $3\times |V|$ equivalence classes. Since each $x_{v,i}$ belongs to a
  distinct equivalence class, there are exactly $3\times |V|$
  equivalence classes. Each $y_v$ belongs to some equivalence class
  $[x_{v,i}]_{\equiv_G}$. Define a color assignment
  $\pi : V \rightarrow \{1,2,3\}$ such that $\pi(v) = i$ iff
  $y_v \in [x_{v,i}]_{\equiv_G}$.  For every $e = uv \in E$, either
  $z_{e, i} \equiv_G x_{u,i}$ or $z_{e, i} \equiv_G x_{v,i}$.  Thus,
  either $\pi(u) \neq i$ or $\pi(v) \neq i$. Therefore, $\pi$ is a
  3-coloring of $H$.
\end{proof}

\section{Policy Mining for Domain and Type Enforcement}
\label{sec:DTE}

We pointed out that IoT devices play both the roles of subject and
object.  In our treatment of domain-based policies, an entity (a
device) is assigned to a domain based on both its characteristics as a
subject and its characteristics as an object.  Generalizing
domain-based policies, the \Dfn{domain and type enforcement} scheme
\cite{Badger-etal:1995} assigns two labels to each entity.  The
\Dfn{domain} label classifies entities based on the kind of access
requests they are allowed to make as subjects, while the \Dfn{type}
label classifies entities according to the kind of accesses they are
allowed to accept as objects.  As we shall see below (Observation
\ref{obs-domain-type} (\ref{obs-domain-type-smaller})), this dual
classification could potentially lead to a smaller number of labels
than one would have needed in an equivalent domain-based policy.  The
result is a cleaner and simpler policy specification.  In this
section, we study the mining of domain and type enforcement policies
from incomplete access logs.


To facilitate the visualization of proof arguments (for
Theorem \ref{thm-DB-is-hard} and \ref{thm-DTEPM-is-hard}), we shift
our representation of the access control matrix from a digraph to its
adjacency matrix.
\begin{definition}
  Suppose we fix the set $\Sigma$ of access rights.
  \begin{itemize}
  \item A digraph $G$ with a vertex set \Entities can be represented
    by an $|\Entities| \times |\Sigma| \times |\Entities|$ boolean
    matrix $M$ so that $(u,a,v) \in E(G)$ iff $M[u,a,v] = 1$.  To
    simplify nomenclature, we refer to $M$ as an \Dfn{access
      control matrix} for the entity set \Entities.
  \item The indistinguishability relation $\equiv$ over 
    \Entities (from Def.~\ref{def:indistinct}) can be easily
    reformulated in terms of access control matrices.
  \item A partially specified digraph $\PSG$ with vertex set \Entities
    can be represented by an
    $|\Entities| \times |\Sigma| \times |\Entities|$ matrix \PSM so
    that $\PSM[u,a,v]$ is $1$ if $(u,a,v) \in E(G)$, $0$ if
    $(u,a,v) \in F(G)$, and $*$ otherwise. We call \PSM a
    \Dfn{partially specified access control matrix} for the entity set
    \Entities.
  \item $M$ is an \Dfn{instantiation} of \PSM iff $M$ can be obtained
    from \PSM by replacing each $*$-entry by either a $0$ or a $1$.
  \end{itemize}
\end{definition}
The following 
formalizes a domain and type enforcement policy.
\begin{definition}
  Suppose $E$ is a set of entities.
  \begin{itemize}
  \item A \Dfn{domain and type enforcement policy} for \Entities is a
    5-tuple $\DTP = (D, T, \delta, \tau, \TBL)$, where $D$ is a finite
    set of \Dfn{domains}, $T$ is a finite set of \Dfn{types},
    $\delta : E \rightarrow D$ is a \Dfn{domain assignment},
    $\tau : E \rightarrow T$ is a \Dfn{type assignment}, and \TBL is a
    $|D| \times |\Sigma| \times |T|$ boolean matrix.  The intention is
    that the policy \DTP grants access request $(u, a, v)$ iff
    $\TBL[\delta(u), a, \tau(v)] = 1$.
  \item A domain and type enforcement policy
    $\DTP = (D, T, \delta, \tau, \TBL)$ \Dfn{enforces} access control
    matrix $M$ whenever the following holds: for every $u,v \in E$ and
    $a \in \Sigma$, $M[u,a,v] = 1$ iff
    $\TBL[\delta(u), a, \tau(v)] = 1$.
  \end{itemize}
\end{definition}
When we attempt to mine a domain and type enforcement policy from
incomplete access logs (\PSM), the mining algorithm can exercise
discretion on whether to interpret a ``don't care'' entry (``$*$'') as
a $1$ or a $0$. We desire an instantiation that leads to the smallest
number of domains and/or types. The decision form of this optimization
problem is the following.

\textbf{Domain and Type Enforcement Policy Mining (DTEPM)}
\begin{itemize}
\item \textbf{Instance:} A positive integer $m$, a set $\Sigma$ of
  access rights, a set \Entities of entities, and a partially
  specified access control matrix \PSM.
\item \textbf{Question:} Is there an instantiation $M$ of \PSM and
  a domain and type enforcement policy
  $\DTP = (D, T, \delta, \tau, \TBL)$ such that
  (a) \DTP enforces $M$ and (b) $\max(|D|,|T|) \leq m$?
\end{itemize}

It turns out that DTEPM is NP-complete (Theorem
\ref{thm-DTEPM-is-hard}).  To establish this result, we adopt a
methodology similar to the study of DBPM: we define equivalence
relations to capture the abstraction of subjects into
domains and the abstraction of objects into types.

\begin{definition}
  Suppose $M$ is an access control matrix for some entity set
  \Entities.
  \begin{itemize}
  \item Given entities $u,v \in \Entities$, we write $u \eqsub v$
    whenever $M[u,a, \linebreak x] = M[v,a,x]$ for every
    $x \in \Entities$ and $a \in \Sigma$.  Symmetrically, we write
    $u \eqobj v$ whenever $M[x,a,u] = M[x,a,v]$ for every
    $x \in \Entities$ and $a \in \Sigma$. We call the two equivalence
    relations \eqsub and \eqobj the \Dfn{row equivalence} and the
    \Dfn{column equivalence} respectively.
  \item We write $r(M)$ to denote the number of equivalence classes
    induced by \eqsub, $c(M)$ to denote the number of equivalence
    classes induced by \eqobj, and $e(M)$ to denote the number
    of equivalence classes induced by $\equiv$.
  \end{itemize}
\end{definition}
If $u \eqsub v$, then as subjects, $u$ and $v$ are allowed to perform
the same set of accesses. Conversely, when $u \eqobj v$, the same set
of accesses can be performed against $u$ and $v$ as objects.  Here are
some observations regarding the above definition.
\begin{observation} \label{obs-domain-type}
  Suppose \Entities is a set of entities,
  $M$ is an access control matrix for \Entities, and
  $\DTP = (D, T, \delta, \tau, \TBL)$ is a domain and type enforcement
  policy that enforces $M$.
  \begin{enumerate}
  \item Like $\equiv$, both $u \eqsub v$ and $u \eqobj v$ can be
    checked in time linear to $|\Entities|$ (assuming $|\Sigma|$ is a
    constant).
  \item $u \equiv v$ iff both $u \eqsub v$ and $u \eqobj v$.
    Thus $r(M) \leq e(M)$ and $c(M) \leq e(M)$. \label{obs-domain-type-smaller}
  \item If $u \not \eqsub v$, then it must be the case that
    $\delta(u) \neq \delta(v)$, or else \DTP would not enforce
    $(E, M)$.  In other words, $|D| \geq r(M)$. \label{obs-domain-type-r}
  \item Applying the same logic to columns, $|T| \geq c(M)$. \label{obs-domain-type-c}
  \item \DTP can be formulated so that $|D| = r(M)$ and $|T| = c(M)$.
    Specifically, we assign a different domain to each
    \eqsub-equiv\-a\-lence class, and a different type to each
    \eqobj-equivalence class. The function $\delta$ (resp.~$\tau$) can
    then be defined to map each entity to the domain (resp.~type) of
    its corresponding \eqsub-equiv\-a\-lence class
    (resp.~\eqobj-equivalence class). By (\ref{obs-domain-type-r}) and
    (\ref{obs-domain-type-c}) above, \DTP enforces $M$ so that
    $\max(|D|, |T|)$ is minimized. \label{obs-domain-type-min-max}
  \end{enumerate}
\end{observation}

The proof of the NP-hardness of DTEPM involves a reduction from the
decision problem below.

\textbf{Domain Bounding (DB)}
\begin{itemize}
\item \textbf{Instance:} A positive integer $m$, a set $\Sigma$ of
  access rights, a set \Entities of entities, and a partially
  specified access control matrix \PSM for \Entities.
\item \textbf{Question:} Does there exist an instantiation $M$ of
  \PSM such that $r(M) \leq m$?
\end{itemize}

\begin{theorem} \label{thm-DB-is-hard}
  DB is NP-complete. This is true even if the
  set $\Sigma$ of access rights is restricted to be a singleton set.
\end{theorem}
\begin{proof}
  Membership in NP is obvious.  We prove NP-hardness by a reduction
  from Graph 3-Colorability \cite[p.~191]{Garey-Johnson:1979} to DB.

  Given undirected graph $G = (V, E)$, we assume, without loss of
  generality, the vertices are $1, 2, \ldots, |V|$, and the edges are
  numbered $1, 2, \ldots, |E|$, such that the $i$'th edge is $u_iv_i$
  where $u_i < v_i$. We denote the three colors by $1$, $2$ and $3$.
  
  The reduction generates from $G$ a DB instance
  $(3, \Sigma, \Entities, \PSM)$, where $\Sigma$ is a singleton set,
  \Entities is a set of size $n = 3 + |V| + 3 \times |E|$, and \PSM is
  a partially specified access control matrix of size
  $n \times 1 \times n$.  To simplify presentation, we treat \PSM and
  its instantiations as 2-dimensional matrices of size $n \times n$ in
  the rest of this proof.

  The rightmost $3 \times |E|$ columns of \PSM are padded with $0$'s.
  We specify below the remaining entries of \PSM.  More specifically,
  there are three types of row in \PSM.

  Rows $1$ to $3$ are the \Dfn{color rows}, so that row $i$
  representing color $i$.  Specifically, $\PSM[i, i] = 1$,
  $\PSM[i, j] = 0$ for $j = \{1,2,3\} \setminus \{i\}$, and
  $\PSM[i, j] = *$ for $j \in \{ 4, 5, \ldots, 3+|V| \}$.

  Rows $4$ to $(3+|V|)$ are the \Dfn{vertex rows}.  Specifically, row
  $(3+v)$ is the row for vertex $v \in V$, such that
  $\PSM[3+v, 3+v] = 1$, and $\PSM[i,j] = *$ for
  $j \in \{ 1, 2, \ldots, 3 + |V| \} \setminus \{ 3+v \}$.

  There are three \Dfn{edge rows} for each undirected edge in $E$.
  Let $\mathit{base}(i)$ be $3+|V|+3(i-1)$.  The edge rows for the
  $i$'th edge $u_iv_i$ are located at rows $\mathit{base}(i) + 1$,
  $\mathit{base}(i) + 2$, and $\mathit{base}(i) + 3$.  For each of the
  three edge rows of $u_iv_i$, in columns $1$ to $(3+|V|)$, except for
  columns $u_i$ and $v_i$, all entries are $*$.  In row
  $\mathit{base}(i) + 1$, columns $u_i$ and $v_i$ store $1$ and $0$
  respectively. In row $\mathit{base}(i) + 2$, columns $u_i$ and $v_i$
  store $0$ and $1$ respectively.  In row $\mathit{base}(i) + 3$,
  columns $u_i$ and $v_i$ store $0$ and $0$ respectively.

  We now demonstrate that $G$ is 3-colorable iff
  $(3, \Sigma, \Entities, \PSM)$ is a positive instance of DB.

  $(\Rightarrow)$ Suppose $\pi : V \rightarrow \{ 1, 2, 3 \}$ is a
  3-coloring of $G$.  We construct an instantiation $M$ of \PSM.  For
  each color $i$, we instantiate the corresponding color row by
  setting $M[i, 3+v] = 1$ if $\pi(v) = i$ and $M[i,3+v] = 0$
  otherwise. Note that the three color rows in $M$ are pairwise
  distinct.  For each vertex $v$, we instantiate the corresponding
  vertex row to the color row $i$ of $M$ if $\pi(v) = i$.  For the
  $i$'th edge $u_iv_i$, we instantiate row $\mathit{base}(i) +1$ to
  the color row $\pi(u_i)$ of $M$, row $\mathit{base}(i) +2$ to the
  color row $\pi(v_i)$ of $M$, and row $\mathit{base}(i) + 3$ to the
  remaining color row of $M$. Therefore,
  $(3, \Sigma, \Entities, \PSM)$ is a positive instance of DB.

  $(\Leftarrow)$ Suppose \PSM has an instantiation $M$ such that all
  rows belong to one of three distinct bit patterns.  Observe that the
  three color rows in \PSM have distinct bit patterns in columns $1$
  to $3$. So the three color rows in $M$ must be the three distinct
  bit patterns to which all other rows in $M$ conform.  Observe also
  that the three edge rows for the $i$'th edge $u_iv_i$ have distinct
  bit patterns at columns $u_i$ and $v_i$. Thus the three edge rows
  must have a one-one correspondence to the three color rows in $M$.

  We construct a color assignment $\pi : V \rightarrow \{ 1, 2, 3 \}$
  as follows. For each isolated vertex $v$ in $G$, set $\pi(v) = 1$.
  Then, for each color $i$, set $\pi(v) = i$ if $M[i, 3+v] = 1$ and
  $v$ is not an isolated vertex in $G$. We argue that (a) every vertex
  receives exactly one color assignment, and (b) no adjacent vertices
  receive the same color assignment. Thus $\pi$ is a 3-coloring for
  $G$.

  To see (a), consider vertex $u$. Either $u$ is an isolated vertex or
  $u$ has a neighbor. In the former case, $\pi(u) = 1$. In the latter
  case, let $v$ be a neighbor $u$. Among the three edge rows for
  edge $uv$, exactly one of them has a $1$ in column $u$. Thus $u$
  receives exactly one color assignment. To see (b), suppose adjacent
  vertices $u$ and $v$ have the same color. Then there is a color row
  with a $1$ in both column $u$ and column $v$, which in turn means
  the same bit pattern occurs in one of the three edge rows for
  $uv$. This is impossible by construction.
\end{proof}
We now return to the main result of this section.
\begin{theorem} \label{thm-DTEPM-is-hard}
  DTEPM is NP-complete. This is true even if $\Sigma$ contains
  only one access right.
\end{theorem}
\begin{proof}
  Membership in NP is obvious.  We demonstrate NP-hard\-ness by a
  reduction from the single access right version of DB.  Given DB
  instance $(m, \Sigma, \Entities, \PSM)$, where $|\Sigma| = 1$ and
  $|\Entities| = n$, our reduction constructs an DTEPM instance
  $(m+2n+1, \Sigma, \Entities', \PSM')$, where \Entities has a size of
  $(3n+1)$ and $\PSM'$ is the following $(3n+1) \times (3n+1)$ partially
  specified access control matrix (again, we treat \PSM, $\PSM'$, and
  their instantiations as 2-dimensional matrices since $\Sigma$ is a
  singleton set):
  \[
    \PSM' =
    \begin{bmatrix}
      \PSM & \textbf{0} \\
      \textbf{0}^T & M^\star
    \end{bmatrix}
  \]
  Here, $\textbf{0}$ is an $n\times (2n+1)$ zero matrix and
  $M^\star$ is the following $(2n+1)\times (2n+1)$ access control matrix.
  \begin{equation*}
    M^\star = 
    \begin{bmatrix}
      1 & 0 & 0 & \dots & 0 & 0 & 1 \\
      0 & 1 & 0 & \dots & 0 & 1 & 0 \\
      0 & 0 & \ddots & \dots & \reflectbox{$\ddots$} & 0 & 0 \\
      \vdots & \vdots & 0 & 1 & 0 & \vdots & \vdots \\
      0 & 0 & 1 & \dots & 1 & 0 & 0 \\        
      0 & 1 & 1 & \dots & 1 & 1 & 0 \\        
      1 & 1 & 1 & \dots & 1 & 1 & 1
    \end{bmatrix}
  \end{equation*}
  Suppose $M$ is an instantiation of \PSM, and $M'$ is the
  corresponding instantiation of $\PSM'$.  Then $r(M') = r(M) + 2n+1$,
  but $c(M') = c(M) + n + 1$.  Thus $r(M') > c(M')$, and
  $\max(r(M'), c(M')) = r(M') = r(M) + 2n+1$.
  Thus $(m, \Sigma, \Entities, \PSM)$ is a positive instance of DB iff
  $(m+2n+1, \Sigma, \Entities', \PSM')$ is a positive instance of
  DTEPM.
\end{proof}

%
%

\begin{table*}[t]
\centering
\resizebox{\textwidth}{!}{%
\begin{tabular}{|lllllllll|}
  \hline
  & \multicolumn{3}{l}{\textbf{Boolean Variables:}}
  & \hspace{0.5cm}
  &
  & \hspace{0.5cm}
  &
  &  \\
  & $x_{i,a,j} \in \{ 0, 1 \}$
  & \hspace{0.5cm}
  &
    \makecell[tl]{$\forall a \in [k]$, $\forall i, j \in [n]$ \\
    for which $\PSM[i,a,j] = *$
    }
  & \hspace{0.5cm}
  &
  & \hspace{0.5cm}
  &
    \makecell[tl]{--- $x_{i,a,j}$ is the boolean instantiation
    of the \\ ``don't care''-entry $\PSM[i,a,j]$. Together these \\variables
  encode an instantiation $M$ of \PSM.
  }
  & \\
  & $y_{i,p} \in \{ 0, 1 \}$
  & \hspace{0.5cm}
  & $\forall i \in [n]$, $\forall p \in [m]$
  & \hspace{0.5cm}
  &
  & \hspace{0.5cm}
  & \makecell[tl]{--- $y_{i,p}$ indicates whether entity $i$
    is assigned to \\ class $p$. Together these
  variables encode a \\protection domain assignment 
  $\pi : [n] \rightarrow [m]$.}
  &  \\
  & $z_{p,a,q} \in \{ 0, 1 \}$
  & \hspace{0.5cm}
  & $\forall a \in [k]$, $\forall p, q \in [m]$
  & \hspace{0.5cm}
  &
  & \hspace{0.5cm}
  & \makecell[tl]{--- $z_{p,a,q}$ indicates whether entities in class
    $p$ may \\
  exercise access right $a$ over entities in class $q$. \\
  Together these variables specify a digraph $H$\\ whose
  vertices are the classes.}
  &  \\
  & $r_{p} \in \{ 0, 1 \}$
  & \hspace{0.5cm}
  & $\forall p \in [m]$
  & \hspace{0.5cm}
  &
  & \hspace{0.5cm}
  & \makecell[tl]{--- $r_p$ indicates whether class $p$ is occupied.}
  &  \\
  & \multicolumn{3}{l}{}
  & \hspace{0.5cm}
  &
  & \hspace{0.5cm}
  &
  &  \\
  & \multicolumn{3}{l}{\textbf{Hard Clauses:}}
  & \hspace{0.5cm}
  &
  & \hspace{0.5cm}
  &
  &  \\
  & $\bigvee\limits_{p \in [m]} y_{i,p}$
  & \hspace{0.5cm}
  & $\forall i \in [n]$
  & \hspace{0.5cm}
  & (1)
  & \hspace{0.5cm}
  & \makecell[tl]{$\pi$ maps each entity to at least $1$ class.}
  &  \\
  & $(\neg y_{i,p} \lor \neg y_{i,q})$
  & \hspace{0.5cm}
  & $\forall i \in [n]$, $\forall p,q \in [m]$, $p < q$
  & \hspace{0.5cm}
  & (2)
  & \hspace{0.5cm}
  & \makecell[tl]{$\pi$ maps each entity to at most $1$ class.} & \\
  & \multicolumn{3}{l}{\makecell[tl]{
    $\forall a \in [k]$, $\forall i, j \in [n]$,
    $\forall p, q \in [m]$:}}
  & \hspace{0.5cm}
  &
  & \hspace{0.5cm}
  & \makecell[tl]{$\pi$ is a strong homomorphism.}
  &  \\
  & \hspace{0.5cm}$(\neg y_{i,p} \lor \neg y_{j,q} \lor \neg
    z_{p,a,q})$
  & \hspace{0.5cm}
  & if $M[i,a,j] = 0$
  & \hspace{0.5cm}
  & (3)
  & \hspace{0.5cm}
  & \makecell[tl]{$z_{p,a,q} \implies 0$}
  &  \\
  & \hspace{0.5cm}$(\neg y_{i,p} \lor \neg y_{j,q} \lor z_{p,a,q})$
  & \hspace{0.5cm}
  & if $M[i,a,j] = 1$
  & \hspace{0.5cm}
  & (4)
  & \hspace{0.5cm}
  & \makecell[tl]{$1 \implies z_{p,a,q}$}
  &  \\
  & \hspace{0.5cm}\makecell[tl]{
    $(\neg y_{i,p} \lor \neg y_{j,q} \lor
    x_{i,a,j} \lor \neg z_{p,a,q})$}
  & \hspace{0.5cm}
  & if $M[i,a,j] = *$
  & \hspace{0.5cm}
  & (5)
  & \hspace{0.5cm}
  & \makecell[tl]{$z_{p,a,q} \implies x_{i,a,j}$}
  &  \\
  & \hspace{0.5cm}\makecell[tl]{$(\neg y_{i,p} \lor \neg y_{j,q} \lor
    \neg x_{i,a,j} \lor z_{p,a,q})$}
  & \hspace{0.5cm}
  & if $M[i,a,j] = *$
  & \hspace{0.5cm}
  & (6)
  & \hspace{0.5cm}
  & \makecell[tl]{$x_{i,a,j} \implies z_{p,a,q}$}
  &  \\
  & $(\neg y_{i,p} \lor r_{p})$
  & \hspace{0.5cm}
  & $\forall i \in [n]$, $\forall p \in [m]$
  & \hspace{0.5cm}
  & (7)
  & \hspace{0.5cm}
  & \makecell[tl]{if $\pi(i) = p$, then class $p$ is occupied.}
  &  \\
  & \multicolumn{3}{l}{}
  & \hspace{0.5cm}
  &
  & \hspace{0.5cm}
  &
  &  \\
  & \multicolumn{3}{l}{\textbf{Soft Clauses:}}
  & \hspace{0.5cm}
  &
  & \hspace{0.5cm}
  &
  &  \\
  & $(\neg r_{p})$
  & \hspace{0.5cm}
  & $\forall p \in [m]$
  & \hspace{0.5cm}
  & (8)
  & \hspace{0.5cm}
  & Maximize the number of unoccupied classes. &  \\
  \hline
\end{tabular}
}%
\caption{The baseline encoding of DBPM}
\label{tab:maxsat_form}
\end{table*}

\section{MaxSAT Encoding of DBPM}

\label{sec:encoding}

Even though DBPM is NP-complete, in practice we still want to be able
to perform mining. In this section, we devise MaxSAT encodings for
DBPM, so that a DBPM instance can be solved by an off-the-shelf MaxSAT
solver. 

\Dfn{MaxSAT} \cite[Ch.~24]{HSAT2021} is an optimization version of SAT
in which a truth assignment is sought for a given CNF formula so
that the number of satisfied clauses is maximized.  \Dfn{Partial
  MaxSAT} is a (generalized) variant of MaxSAT in which each clause is
either a hard clause or a soft clause.  The optimization objective is
to find a truth assignment that satisfies all the hard clauses but
maximizes the number of satisfied soft clauses.

\subsection{Input Instance, Parameters, and Notations}
\label{sec:instance}

Suppose we are given a partially specified access control matrix \PSM,
which is a $n \times k \times n$ matrix, where $k$ is the number of
access rights and $n$ is the number of entities.  Each entry is either
a boolean value ($0$ or $1$) or a wildcard symbol (`$*$').  Our goal
is to find an instantiation $M$ of \PSM such that the number of
$\equiv$-equivalence classes is minimized.

Our assumption below is that we already know a certain upper bound for
the number of equivalence classes in the optimal instantiation $M$. We
denote that upper bound by $m$.  Our MaxSAT encodings create $m$
placeholders for equivalence classes, and attempt to assign
indistinguishable entities into the same placeholder. For brevity, we
call such a placeholder a ``\Dfn{class}.''

Given positive integer $x$, we write $[x]$ to denote
$\{1, \ldots, x\}$.

\subsection{The Baseline Encoding}
Table \ref{tab:maxsat_form} shows a partial MaxSAT encoding for the
optimization version of DBPM. We dub this the \Dfn{baseline encoding
  (BE)}.  The boolean variables capture a number of choices: (a) an
instantiation $M$ of \PSM (represented by $x_{i,a,j}$), (b) a
domain-based policy consisting of a protection domain assignment $\pi$
(represented by $y_{i,p}$) and a digraph $H$ (represented by
$z_{p,a,q}$), and (c) a subset of classes known as ``occupied''
classes (represented by $r_p$).  Hard clauses (1) and (2) ensure that
$\pi$ is a function.  Hard clauses (3)--(7) ensure that $\pi$ is a
strong homomorphism.  Hard clause (8) identifies the ``occupied''
classes to be those in the range of $\pi$.  The optimization objective
is to minimize the number of ``occupied'' classes.

\subsection{Optimizing the ``At Most One'' Constraint}
\label{sec:at-most-one}

\setcounter{equation}{8} 

Hard clause (2) ensures that the function $\pi$ has \emph{at most one}
image for each input value.  The following are two alternative ways to
encode the ``at most one'' constraint. These two optimization ideas
are mutually exclusive: only one can be applied at a time.

\textbf{1. Cardinality Constraint Optimization (CC).}
Hard clauses (1) and (2) encode the following
\Dfn{cardinality constraint} \cite[Ch.~28]{HSAT2021}, in which
boolean variables are interpreted as integers ($0$ or $1$):
\begin{equation}
  \left(\sum_{p \in [m]} y_{i,p} = 1\right)
\end{equation}
(1) and (2) together form the so called \Dfn{pairwise encoding} of
cardinality constraints, in which (2) expands to $O(nm^2)$
clauses. More efficient encodings of cardinality constraints
exist. For example, in our experiments (\S \ref{sec:evaluation}), we
use the ladder encoding \cite[Ch.~2]{HSAT2021}, which introduces only
linearly many new clauses at the expense of linearly many new
variables.

\textbf{2. Non-Functional Optimization (NF).}  Impose only hard clause
(1) and \emph{omit} hard clause (2) in the baseline encoding.
Effectively, we allow each $i$ to be assigned to more than one class.
Let us suppose $i$ gets assigned to multiple classes.  Then the strong
homomorphism constraints will ensure that the entities assigned to
these classes are equivalent ($\equiv$).  The optimization objective
will thus prefer ``merging'' these classes rather than leaving them
separate. The result of the optimization is that equivalent entities
will get assigned to the same class. So we get the effect of imposing
hard clause (2) without needing to add any extra constraints, thus
resulting in faster performance.  (This optimization idea is inspired
by the multivalued direct encoding \cite{Selman-etal:1992} and minimal
support encoding \cite{Argelich-etal:2008} of Constraint Satisfaction
Problems as SAT instances.)

\subsection{Symmetry Breaking: Permutations}
With the baseline encoding, the solution space has lots of symmetries.
Suppose there is an optimal solution in which $\pi$ maps the entities
to $s$ of the $m$ classes. Then permuting the $s$ occupied classes,
and adjusting $\pi$ accordingly, will produce another optimal solution.
Similarly, if we select a different $s$-subset of the $m$ classes to
be the range of $\pi$, and adjust $\pi$ accordingly, then we obtain
another optimal solution.  Thus there is a total of
$s! \times \binom{m}{s}$ symmetries to each optimal solution.  In this
subsection, we outline an optimization technique for removing the
symmetries arisen from the permutation of the occupied classes,
thereby reducing the the number of symmetries by a factor of $s!$.
(The next subsection considers the symmetries arisen from
$s$-subsets.)

The idea of this optimization technique is to force the MaxSAT solver
to search for one specific permutation of the $s$ occupied classes,
rather than considering all of them.  For each class $p$, let $\min_p$
be smallest entity index assigned to class $p$.  Essentially, we look
for a mapping $\pi$ such that the $\min_p$ values for the occupied
classes $p$ are sorted: i.e., $\min_p < \min_q$ iff $p < q$. To
realize this idea, we devised two encodings. They are mutually
exclusive: only one can be applied at a time.

We call the first encoding \Dfn{Feasible Mins (FM)},
which introduces the
following variables and clauses.

\textbf{Boolean Variables:}
\begin{itemize}
\item $l_{i,p}$ for $i \in [n]$ and $p \in [m]$.
  The variable $l_{i,p}$ asserts that entity $i$ is
  the lowest-indexed entity that has been mapped by $\pi$ to
  class $p$.
\end{itemize}

\textbf{Hard Clauses:} 
\begin{itemize}
\item \emph{Mins are sorted.} For $p,q \in [m]$ such that $p < q$,
  and for $i,j \in [n]$ such that $j \leq i$,
  \begin{equation} \label{eqn-mins-sorted-2}
    (\lnot l_{i,p} \lor \lnot l_{j,q})
  \end{equation}

\item \emph{Min is minimum.} For $i,j \in [n]$ such that $i < j$, and
  for $p \in [m]$,
  \begin{equation} \label{eqn-min-is-minimum-2}
    (\lnot y_{i,p} \lor \lnot l_{j,p})
  \end{equation}

\item \emph{Min is selected.} For $i \in [n]$ and $p \in [m]$,
  \begin{equation} \label{eqn-min-is-selected-2}
    (\lnot l_{i,p} \lor y_{i,p})
  \end{equation}

\item \emph{Feasible mins.} For $i \in [n]$ and $p \in [m]$,
  \begin{equation} \label{eqn-feasible-mins-2} 
    (\lnot y_{i,p} \lor l_{1,p} \lor l_{2,p} \lor \cdots \lor l_{i,p})
  \end{equation}
\end{itemize}

A second encoding for eliminating permutation symmetries is called
\Dfn{Min Domain (MD)}, which is obtained from FM by replacing hard
clause \eqref{eqn-feasible-mins-2} by the following hard clause.
\begin{itemize}
\item \emph{Min domain.} For every $p \in [m]$,
  \begin{equation} \label{eqn-min-domain-2}
    (\lnot r_p \lor l_{1,p} \lor l_{2,p} \lor \cdots \lor l_{n,p})
  \end{equation}
\end{itemize}

MD generates significantly fewer clauses than FM.  Specifically,
\eqref{eqn-feasible-mins-2} in FM is expanded into $n\times m$
clauses, each of size $O(n)$, whereas \eqref{eqn-min-domain-2} in MD
corresponds to $m$ clauses, each of size $O(n)$.  However, FM
propagates constraints more directly than the MD.


\subsection{Symmetry Breaking: Combinations}

An additional optimization can be performed to eliminate the
symmetries arising from the choice of an $s$-subset of classes to be
used as the range of function $\pi$.  Doing so reduces the number of
symmetries by a factor of $\binom{m}{s}$.

The optimization idea is actually quite simple: we prefer to assign
entities to classes with lower indices before assigning entities to
higher-indexed classes. More specifically, we only assign entities to
those classes with indices $1, 2, \ldots, s$.  Of course, we do not
know beforehand what $s$ is, as that is up to the MaxSAT solver to
figure out.  This idea is realized in an encoding called
\Dfn{Lower-Indexed Classes (LI)}, which imposes the following hard
clauses.
\begin{itemize}
\item \emph{Prefer lower-indexed classes.}  For $p \in [m-1]$,
  \begin{equation}
    (r_p \lor \lnot r_{p+1})
  \end{equation}
\end{itemize}
The clause above requires that if class $p$ is unoccupied, then class
$p+1$ must be unoccupied as well. The effect is that classes with
lower indices will be occupied first.

\section{Experimental Evaluation} \label{sec:evaluation}

An experiment was conducted with the goals of: (1) evaluating the
relative performance of the encodings BE, CC, NF, FM, MD, and LI, as
well as their combinations; (2) assessing if MaxSAT solving is
adequate for solving DBPM in practice.

\subsection{Experiment Setup}

The experiments were performed on one node of a high-performance
computing cluster.  The node has $80$ Intel Xeon Gold 6148 @ 2.40GHz
CPUs and runs Rocky Linux 8.7 with $3000$GB memory.  The experiments
were run as single-node jobs and each job allocated $1$ CPU with
$256$GB memory.  The RC2 algorithm \cite{Ignatiev2019} was used for
solving MaxSAT problems.  RC2 participated in the MaxSAT Evaluations
2018 and 2019, and it was ranked first in two complete categories:
unweighted and weighted.  The implementation was provided by the PySAT
\cite{Pysat2018} library.


\subsection{Benchmarking Method}

We benchmarked the performance of 6 encodings: (i) BE, (ii) BE+CC,
(iii) BE+NF, (iv) BE+NF+FM, (v) BE+NF+MD, and (vi) BE+NF+MD+LI. To
that end, we created a benchmark suite consisting of a collection
of DBPM instances.

Given parameters $m^*$ and $n$, we generated a DBPM instance using the
following steps: (1) Randomly generate a digraph $H$ with $m^*$
vertices. $|\Sigma|$ is set to $1$. Each triple $(u,a,v)$ has a $0.5$
probability of being an edge.  (2) Generate a digraph $G$ with $n$
vertices and a protection domain assignment $\pi$.  This is achieved
by evenly assigning the $n$ vertices of $G$ to the $m^*$ vertices of
$H$, and then replicating the edges of $H$ in $G$ to ensure $\pi$ is a
strong homomorphism. The mapping $\pi$ is then discarded. (3) Generate
a partially specified access control matrix \PSM.  To do that, $G$ is
first turned into an access control matrix $M$.  Then we sample 10\%
of the matrix entries to be turned into ``don't care'' ($*$), thereby
obtaining \PSM.

For each $m^* \in \{2, 4, 6, 8, 10\}$, and for each
$n \in \{100, 200, \ldots, 1000\}$, we generated $6$ DBPM instances,
resulting in a total of $300$ DBPM instances as our benchmark suite.
We then applied each of the 6 encodings to solve each of the DBPM
instances.  We used $m = 2\times m^*$ as the estimated upper bound for
the number of equivalence classes (\S \ref{sec:instance}).  We also
set a timeout limit of $5$ minutes: i.e., we terminated the solving of
an instance by an encoding when that instance could not be solved by
that encoding within the timeout limit. When an instance could be
solved by an encoding within the timeout limit, we measured the
execution time using the \texttt{time()} function from the time module
of Python.

\begin{figure*}
\centering
\begin{tikzpicture}
\begin{axis}[
  ylabel=Execution time (seconds),
  xlabel=Number of instances,
  width= 0.9*\textwidth,
  height= 0.32*\textwidth,
  xmin=0,
  xmax=300,
  ymin=0,
  ymax=20000,
  scaled y ticks = false,
  grid=both,
  legend entries={
                  BE,
                  BE+CC,
                  BE+NF,
                  BE+NF+FM,
                  BE+NF+MD,
                  BE+NF+MD+LI},
  every axis legend/.append style={at={(0, 1)}, anchor=north west, outer xsep=5pt, outer ysep=5pt, legend columns=3, nodes={scale=1.0}}
    ]
\addplot[mark=|, color=black] table {BE.csv};
\addplot[mark=x, color=brown] table {BE_CC.csv};
\addplot[mark=o, color=violet] table {BE_NF.csv};
\addplot[mark=|, color=cyan] table {BE_NF_FM.csv};
\addplot[mark=x, color=blue] table {BE_NF_MD.csv};
\addplot[mark=o, color=red] table {BE_NF_MD_LI.csv};
\end{axis}
\end{tikzpicture}
\caption{Cactus plot for the results.\label{fig:cactus}}
\end{figure*}
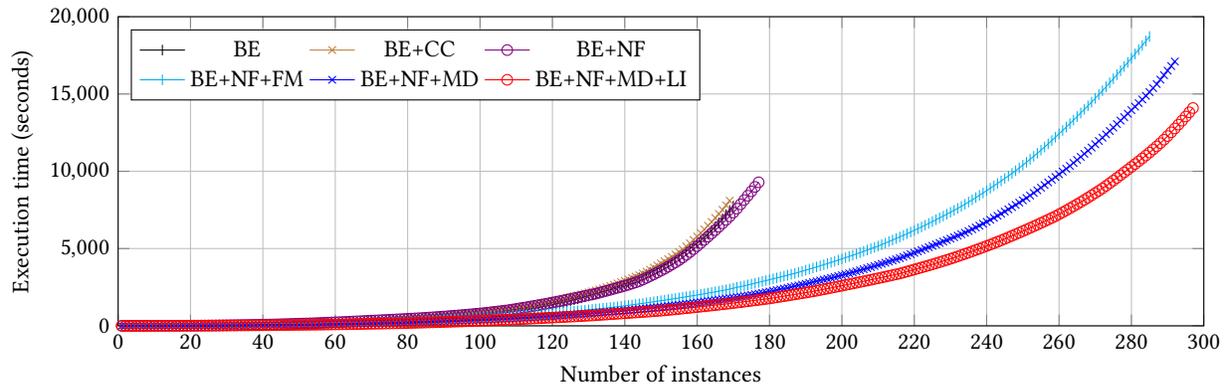

\begin{table}
  \begin{tabular}{|l|c|c|} \hline
    \textbf{Encoding}
    & \textbf{Instances Solved}
    & \textbf{Total Time (sec)} \\ \hline \hline
    BE          & 170 & 7,746   \\ \hline
    BE+CC       & 169 & 8,097   \\ \hline
    BE+NF       & 177 & 9,286   \\ \hline
    BE+NF+FM    & 285 & 18,719  \\ \hline
    BE+NF+MD    & 292 & 17,103  \\ \hline
    BE+NF+MD+LI & 297 & 14,092  \\ \hline
  \end{tabular}
  \caption{Number of instances solved and total execution time for the 6
  encodings\label{tbl:results-in-table}}
\end{table}

\subsection{Results}

Table \ref{tbl:results-in-table} lists, for each encoding, the number
of DBPM instances solved within the timeout limit and the total
execution time to solve all those instances.  Figure \ref{fig:cactus}
visualizes the results in a \Dfn{cactus plot}, which is now the
standard way by which the SAT community depicts experimental results
\cite{Brain2017}.
The X-axis is the number of solved instances, and the Y-axis is the
total execution time (in seconds) to solve those instances. More
specifically, suppose an encoding solved $N$ of the $300$ DBPM
instances. We sorted the execution time of the $N$ instances:
$t_1 \leq t_2 \leq \cdots \leq t_N$.  Then we computed
$T_i = t_1 + \cdots t_2 + \cdots + t_i$. Lastly, we plotted
$(T_i, i)$, for $1 \leq i \leq N$, to depict the performance of that
encoding. In other words, a point $(T_i, i)$ in the cactus plot tells
us that the encoding in question was able to solve $i$ instances
within $T_i$ seconds. We repeated this for all the $6$ encodings to
obtain Figure \ref{fig:cactus}.

There are three takeaways from the results.
\begin{enumerate}
\item \textbf{The most efficient encoding is BE+NF+MD+LI.}  This
  encoding combines the optimization of the ``at most one'' constraint
  via NF, the breaking
  of permutation symmetries via MD, and the breaking
  of combination symmetries via LI. It demonstrates that
  each of the three types of optimization is important in the
  crafting of a MaxSAT encoding for DBPM.
\item \textbf{Symmetry breaking produces a more noticeable speed\-up
    than optimizing the ``at most one'' constraint.}  While BE+NF
  solved more instances than BE, significant performance improvement
  came from the use of FM, MD, and LI.  Comparing the two ways to
  optimize ``at least one'', NF has a slight performance advantage
  over CC. Note that BE+NF generates fewer clauses than BE and BE+CC.
\item \textbf{MD is more efficient than FM in breaking permutation
    symmetries.}  This speedup is likely due to MD generating a much
  smaller number of clauses than FM ($m$ clauses vs $m\times n$
  clauses).
\end{enumerate}

\section{Related Work} \label{sec:related-work}
\subsection{Graph Homomorphism}
Given directed graphs $G$ and $H$, a homomorphism from $G$ to
$H$ is a function $h : V(G) \rightarrow V(H)$ so that
$(u, v) \in E(G)$ implies $(h(u), h(v)) \in E(H)$.  While strong
homomorphism (\S \ref{sec:algorithm}) preserves both adjacency and
non-adjacency, homomorphism preserves only adjacency.  The
body of literature behind the study of graph homomorphism is
surveyed by Hell and Ne\v{s}et\v{r}il \cite{Hell2004}.  There are
clear parallels between the theory we developed in \S
\ref{sec:algorithm} concerning digraph strong homomorphism and the
theory of graph homomorphism. In the following, we point out the
parallels and identify the novelty of our contributions.

In the theory of graph homomorphism \cite[\S 1.6 and \S
2.8]{Hell2004}, a \Dfn{retraction} is a homomorphism $h$ from a
directed graph $G$ to a proper subgraph $H$ of $G$, so that $h(v) = v$
for $v \in V(H)$.  In this case, $H$ is called a \Dfn{retract} of
$G$. A directed graph $G$ that has no retraction is called a
\Dfn{core}. A core is unique up to isomorphism.

The theory of strong homomorphism as developed in \S
\ref{sec:algorithm} bears clear parallels to the theory of retraction
and core. Given digraphs $G$ and $H$, a surjective strong homomorphism
is analogous to a retraction, and a summary is analogous to a core.
Nevertheless, digraph strong homomorphism and graph homomorphism are
distinct concepts. We contrast the two along three dimensions: (a)
logical relationship, (b) complexity, and (c) characterization.

(a) Logically, every strong homomorphism is a homomorphism, but not vice
versa. Therefore, every core is a summary, but not every summary is a
core.

(b) Complexity-wise, core recognition is in general
$\mathit{NP}$-complete \cite{Hell1992}, but the recognition and
construction of summary are both in $\mathit{P}$ (Algorithm
\ref{algo:offline}).  We conjecture that testing for surjective strong
homomorphism is in $\mathit{NPI}$, which contains problems that are
commonly believed to be harder than those in $\mathit{P}$ but not as
hard as $\mathit{NP}$-complete problems. More specifically, we
conjecture that surjective strong homomorphism testing is
$\mathit{GI}$-complete, meaning the former problem is equivalent
in hardness to graph isomorphism testing.
In contrast, testing graph
homomorphism is $\mathit{NP}$-complete \cite{Hell2004}.  All this
suggests that graph homomorphism and digraph strong homomorphism are
two distinct computational phenomena.

(c) Another novelty is our characterization of digraph summary using
an equivalence relation (i.e., indistinguishability).  This
characterization is the foundation of the efficient design of
\textsc{Summarize} (Algorithm \ref{algo:offline}).
No analogous characterization is available for graph homomorphism. That explains the difference in complexity between core recognition and summary recognition/construction.

\subsection{Role Mining and RBAC Policy Analysis}

The advantages of Role-Based Access Control (RBAC) depend on the
complex structure of role assignment, role authorization, and role
permissions.  Role Mining (RM) \cite{Vaidya2007} is an automatic
method to find appropriate roles for organizational needs.
Intuitively, a traditional ACM can be cast as a corresponding RBAC
policy by using matrix decomposition.  User-Permission Assignments
($\mathit{UPA}$), User-role Assignments ($\mathit{UA}$), and
role-Permission Assignments ($\mathit{PA}$) can be represented as
Boolean matrices.  Finding the appropriate set of roles is equivalent
to decomposing $\mathit{UPA}$ into two suitable Boolean matrices:
$\mathit{UA}$ and $\mathit{PA}$.

Mitra et al. \cite{Mitra2016} investigated and analyzed the related
literature in terms of RM.  It proved that the RM problem and its
variants are all NP-hard problems.  For example, let $m$ be the number
of users, $n$ be the number of permissions, and $r$ be the number of
roles, then an $\mathit{UA}$ is a matrix of size $m \times r$, a
$\mathit{PA}$ is a matrix of size $r \times n$ and an $\mathit{UPA}$
is a matrix of size $m \times n$.  The Edge RM problem aims to
minimize $|\mathit{UA}| + |\mathit{PA}|$, where an $\mathit{UA}$ and a
$\mathit{PA}$ are decomposed from the given $\mathit{UPA}$;
$|\mathit{UA}|$ is the total number of roles assigned to all the users
and $|\mathit{PA}|$ is the total number of permissions included in all
the roles.  Xu and Scott \cite{Xu2015} also proved the Attribute-Based
Access Control (ABAC) policy mining problem is NP-hard by reducing
Edge RM to it.

Wickramaarachchi \emph{et al.}~\cite{Wick2009} encoded the User Authorization
Query (UAQ) problem to MaxSAT, and demonstrated their SAT-solving
approach is effective in solving the UAQ problem.

Chen and Crampton \cite{Chen2009} proved that a variation of the set
cover problem, namely, the minimal cover problem is NP-hard. They then
explored the connections between the Inter-Domain Role Mapping (IDRM)
problem and the minimal cover problem.  They proved that the exact
IDRM decision problem is NP-complete, and the exact IDRM optimization
problem is also NP-hard.

Benedetti and Mori \cite{Benedetti2018} encoded their Role Maintenance
problem to MaxSAT and solved it by both complete and incomplete MaxSAT
solvers. A solver is said to be complete when it finds the optimal
solution (if one exists).  A solver is incomplete when it quickly
finds a good enough but possibly suboptimal solution.

Crampton et al. \cite{Crampton2022} proposed the Generalized Noise
Role Mining (GNRM) problem.  Let $r$ be the number of columns in
$\mathit{UA}$ and the number of rows in $\mathit{PA}$.  Let $k$ be the
Hamming distance between $\mathit{UPA}$ and the composition of
$\mathit{UA}$ and $\mathit{PA}$.  They proved that GNRM is
fixed-parameter tractable with parameter $r + k$.  They also expressed
the GNRM problem as a constraint satisfaction problem and solved it by
the CP-SAT solver.

Our work differs from existing work in the following ways: (1) DBPM
(resp.~DTEPM) concerns protection domains (resp.~domains and types)
rather than roles.  (2) We characterized digraph summary in terms of
strong homomorphism and the indistinguishability relation.  Role
mining, however, can be characterized by matrix decomposition.  (3) We
explored different encodings of DBPM in order to speedup MaxSAT
solving.  Such sophisticated optimization techniques were not
attempted in previous work in the literature of RBAC policy analysis.

\section{Conclusion and Future Work} \label{sec:conclusion}

This work studies the mining of domain-based policies from access
logs.  We began by characterizing domain-based policy mining as the
search for digraph summaries.  We demonstrated that, with incomplete
access logs, the general problem of mining domain-based policies can
be seen as a graph sandwich problem, resulting in NP-completeness.
Our experimental results, however, suggest that it is possible to mine
domain-based policies relatively efficiently using MaxSAT solvers,
especially when sophisticated optimization techniques are employed.

We are interested in a number of future work.  (1) In this work, we
interpreted Occam's Razor to mean the minimization of domains.  In
future work, we would like to consider alternative optimization
objectives.  For example, if we know that a small number of ``don't
care'' entries will turn into boolean values in the future, but we do
not know what boolean values they will turn into, how do we construct
domain-based policies that are resilient to future gain of knowledge?
(2) Rather than employing MaxSAT solvers for mining domain-based
policies, we want to investigate if other constrained-optimization
techniques would yield better performance.  (3) We aspire to design
efficient approximation algorithms for solving the optimization
version of DBPM.

\bibliographystyle{abbrv}
\bibliography{references}

\appendix

\section{Lemmas and Proofs}
\label{app:proofs}

\subsection{Proof of Proposition \ref{prop:surjective}}
\label{app:surjective}

\begin{proof}
  Suppose $\pi$ is not surjective. Then there is a proper subset
  $U = \mathit{range}(\pi) \subset V(H)$ such that $G$ is strongly
  homomorphic to $H[U]$, contradicting condition (b) of
  Def.~\ref{def:summary}.
\end{proof}

\subsection{Lemmas about Summaries}
\label{app:summary}

The lemma below shows that every summary is isomorphic to an induced
subgraph of the original digraph.
\begin{lemma} \label{lem:isomorphic}
  If digraph $G$ is strongly homomorphic to digraph $H$ via a
  surjection, then $H$ is isomorphic to an induced subgraph $G[U]$ of
  $G$, where $U \subseteq V(G)$.
\end{lemma}
\begin{proof}
  Suppose the surjective function $\pi : V(G) \rightarrow V(H)$ is a
  strong homomorphism from $G$ to $H$.  For every $v \in V(H)$,
  arbitrarily select a vertex $u_v \in V(G)$ such that $\pi(u_v) = v$.
  Since $\pi$ is surjective, such a selection always exists.  Let
  $U = \{\, u_v \mid v \in V(H) \,\}$.  Since $\pi$ is a strong
  homomorphism, $G[U]$ is isomorphic to $H$ via the bijection
  $\pi|_U$, that is, the restriction of $\pi$ to $U$.
 \end{proof}

The following lemma provides alternative characterizations of
irreducible digraphs.
\begin{lemma}[Characterization of Irreducible Digraphs] \label{lem:char-irred-digraphs}
  Suppose $G$ is a finite digraph. The following statements are
  equivalent:
  \begin{enumerate}
  \item $G$ is irreducible.
  \item $G$ is not strongly homomorphic to any of its proper
    subgraphs.
  \item Every strong homomorphism from $G$ to another digraph is
    injective.
  \end{enumerate}
\end{lemma}
\begin{proof}
  \emph{1) $\Rightarrow$ 2).}  Suppose $G$ is irreducible. Then $G$ is
  its own summary.  By Definition \ref{def:summary}, $G$ is not
  strongly homomorphic to a proper subgraph of its own summary,
  namely, $G$ itself.
  
  \emph{2) $\Rightarrow$ 3).}  Suppose $G$ is not strongly homomorphic
  to any of its proper subgraphs.  Consider a digraph $H$ and a strong
  homomorphism $h$ from $G$ to $H$. Let $W = \Ran{h}$. Then $h$ is a
  surjective strong homomorphism from $G$ to $H[W]$. By Lemma
  \ref{lem:isomorphic}, there is an isomorphism $g$ from $H[W]$ to an
  induced subgraph $G[U]$ of $G$, where $U \subseteq V(G)$.  Now,
  $g \circ h$ is a strong homomorphism from $G$ to its subgraph
  $G[U]$. $G[U]$ cannot be a proper subgraph of $G$, meaning
  $U = V(G)$. This can only be true because $h$ is injective.
  
  \emph{3) $\Rightarrow$ 1).}  Suppose every strong homomorphism from
  $G$ to another digraph is injective.  Suppose digraph $H$ is a
  summary of $G$. By Proposition \ref{prop:surjective}, there is a
  surjective strong homomorphism $h$ from $G$ to $H$.  But $h$
  has to be injective, meaning $h$ is an isomorphism.
\end{proof}

\subsection{Proof of Proposition \ref{prop:irreducible}}
\label{app:irreducible}

\begin{proof}
  \emph{($\Rightarrow$)} Suppose $g$ is a surjective strong
  homomorphism from $G$ to $H$ and $H$ is a summary of $G$.  We claim
  that $H$ is not strongly homomorphic to any of its proper
  subgraph. Suppose that is not the case, and there is a strong
  homomorphism $h$ from $H$ to a proper subgraph $H'$.  Then the
  function $h \circ g$ is a strong homomorphism from $G$ to $H'$,
  contradicting the fact that $H$ is a summary of $G$. The claim is
  therefore valid.  By Lemma \ref{lem:char-irred-digraphs}, $H$ is
  irreducible.
  
  \emph{($\Leftarrow$)} Suppose there is a surjective strong
  homomorphism from $G$ to $H$ and $H$ is irreducible.  Then by Lemma
  \ref{lem:isomorphic}, $H$ is isomorphic to an induced subgraph
  $G[U]$ of $G$, where $U \subseteq V(G)$. Let $g$ be this
  isomorphism.  We claim that there is no strong homomorphism from $G$
  to a proper subgraph of $H$. Suppose otherwise, and $h$ is a strong
  homomorphism from $G$ to a proper subgraph $H'$ of $H$. Then
  $h \circ g$ is a strong homomorphism from $H$ to one of its proper
  subgraphs, contradicting the irreducibility of $H$. Thus the claim
  holds, and $H$ is a summary of $G$.
\end{proof}

\subsection{Proof of Proposition \ref{prop:unique}}
\label{app:unique}

\begin{proof}
  Suppose digraphs $H$ and $H'$ are summaries of $G$. By Propositions
  \ref{prop:surjective} and \ref{prop:irreducible}, $H$ and $H'$ are
  both irreducible.  We show that $H$ and $H'$ are isomorphic to each
  other.

  Suppose $h$ and $h'$ are surjective strong homomorphisms from $G$ to
  $H$ and from $G$ to $H'$ respectively.  Lemma \ref{lem:isomorphic}
  implies there is a subset $U$ of $V(G)$ such that $H$ is
  isomorphic to $G[U]$.  Let $g$ be this isomorphism. Consider the
  function $f = h' \circ g$. Function $f$ is a strong homomorphism (as
  it is the composition of a strong homomorphism with an isomorphism).
  In addition, $f$ is injective by Lemma
  \ref{lem:char-irred-digraphs}.
  
  We claim that $f$ is surjective as well. We derive a contradiction
  in case the claim is not true.  Suppose $f$ is not surjective,
  meaning $\Ran{f} \subset V(H')$. Then $f \circ h$ is a strong
  homomorphism from $G$ to a proper subgraph of $H'$, contradicting
  the fact that $H'$ is a summary of $G$.
  
  In summary, $f$ is a strong homomorphism from $H$ to
  $H'$ that is both surjective and injective, meaning
   $f$ is an isomorphism.
\end{proof}

\subsection{Lemmas about Indistinguishability}
\label{app:indistinguishability}

The following lemma offers an alternative characterization of
indistinguishability that is more elegant than Def.~\ref{def:indistinct}.
\begin{lemma}
  \label{lem:indistinct} 
  Given a digraph $G$, two vertices $u$ and $v$ are indistinguishable
  ($u \equiv_G v$) if and only if the following conditions hold for
  every $a \in \Sigma$ and every $x \in V(G)$,
    \begin{align}
(u, a, x) \in E(G) & \iff (v, a, x) \in E(G)\label{eqn:indistinct:left}\\
(x, a, u) \in E(G) & \iff (x, a, v) \in E(G)\label{eqn:indistinct:right}
    \end{align}
\end{lemma}
\begin{proof}
  The ``if'' direction is obvious. We demonstrate the ``only if''
  direction below.

  Suppose $u \equiv_G v$. Consider $a \in \Sigma$ and $x \in
  V(G)$. There are two cases. 
  \begin{itemize}
  \item \textbf{Case $x \in \{u,v\}$.} Then Condition 1 of Definition
    \ref{def:indistinct} ensures that the following edges either all
    belong to $E(G)$ or all absent from $E(G)$:
    $(u, a, u)$, $(u, a, v)$, $(v, a, u)$, $(v, a, v)$.
    Thus the following statements hold:
    \begin{itemize}
    \item $(u, a, x) \in E(G)$ iff $(v, a, x) \in E(G)$, and
    \item $(x, a, u) \in E(G)$ iff $(x, a, v) \in E(G)$.
    \end{itemize}
  \item \textbf{Case $x \not\in \{u,v\}$.}  Then Condition 2 of
    Definition \ref{def:indistinct} guarantees the following:
    \begin{itemize}
    \item $(u, a, x) \in E(G)$ iff $(v, a, x) \in E(G)$, and
    \item $(x, a, u) \in E(G)$ iff $(x, a, v) \in E(G)$.
    \end{itemize}
  \end{itemize}
\end{proof}

The following definition and lemma offer a shorthand for
articulating arguments concerning indistinguishability.
\begin{definition}
  \label{def:adj}
  Given a digraph $G$ and vertices $u, v \in V(G)$, we use the term
  \Dfn{adjacency from $u$ to $v$}, denoted $\mathit{adj}_G(u, v)$, to
  refer to the set
  $\{\, +a \,\mid\, (u,a,v) \in E(G) \,\} \cup \{\, -a \,\mid\, (v, a,
  u) \in E(G) \,\}$. Vertices $u$ and $v$ need not be distinct, and we
  can drop the subscript $G$ when there is no ambiguities about which
  digraph we are concerned with.
\end{definition}
The lemma below follows immediately from Lemma \ref{lem:indistinct}.
\begin{lemma}
  \label{lem:indistinguishability-in-terms-of-adj}
  Given a digraph $G$ and vertices $x,y \in V(G)$, $x \equiv_G y$ iff
  $\mathit{adj}_G (x, z) = \mathit{adj}_G (y, z)$ for every
  $z \in V(G)$, iff $\mathit{adj}_G (z, x) = \mathit{adj}_G (z, y)$
  for every $z \in V(G)$.
\end{lemma}

The following lemma states that a strong homomorphism only maps
equivalent vertices to the same image. In other words, vertices
with distinct images are inequivalent.
\begin{lemma} \label{lem:sh_inverse} Suppose $\pi$ is a strong
  homomorphism from $G$ to $H$.  For $u,v \in V(G)$, if
  $\pi(u) = \pi(v)$ then $u \equiv_G v$.
\end{lemma}
\begin{proof}
  For brevity, we write $\equiv$ for $\equiv_G$.  Suppose $\pi$ is a
  strong homomorphism from $G$ to $H$.  Suppose further that
  $\pi(u) = \pi(v)$ for vertices $u, v \in V(G)$.  By way of
  contradiction, let us assume that $u \not\equiv v$.
  By Lemma \ref{lem:indistinct}, there exists $a \in \Sigma$ and
  $w \in V(G)$ such that one of the following four conditions holds:
  \begin{gather}
    (u, a, w) \in E(G) \land (v, a, w) \notin E(G) \label{eqn:pf:contra}\\
    (u, a, w) \notin E(G) \land (v, a, w) \in E(G) \\
    (w, a, u) \in E(G) \land (w, a, v) \notin E(G) \\
    (w, a, u) \notin E(G) \land (w, a, v) \in E(G)        
  \end{gather}
  We demonstrate below \eqref{eqn:pf:contra} leads to a
  contradiction. Similar arguments can be applied to show that the rest
  of the cases lead to contradictions as well.

  Suppose \eqref{eqn:pf:contra} holds. Then $(u, a, w) \in E(G)$
  implies $(\pi(u), a, \pi(w)) \in E(H)$. Similarly,
  $(v, a, w) \notin E(G)$ implies $(\pi(v), a, \pi(w)) \notin
  E(H)$. But then $\pi(u) = \pi(v)$, meaning the same edge
  $(\pi(u), a, \pi(w)) = (\pi(v), a, \pi(w))$ is both in $E(H)$ and
  not in $E(H)$. This is a contradiction.
\end{proof}

\subsection{Proof of Proposition \ref{prop:equiv_relation}}

\begin{proof}
  This proposition follows readily from either Lemma
  \ref{lem:indistinct} or Lemma
  \ref{lem:indistinguishability-in-terms-of-adj}.
\end{proof}

\subsection{Proof of Theorem \ref{thm:offline-correctness}}
\label{app:offline-correctness}

\begin{proof}
  Consider the function $f: V(G) \rightarrow V(\Sum{G})$ such that
  $f(v) = [v]_\equiv$. We demonstrate two facts: (a) $f$ is a
  strong homomorphism; (b) $G$ is not strongly homomorphic to any
  proper subgraph of \Sum{G}.

  \emph{(a) $f$ is a strong homomorphism.}  Consider vertices
  $u, v \in V(G)$ and access right $a \in \Sigma$.  Suppose
  $(u,a,v) \in E(G)$. Then by definition of \Sum{G},
  $(f(u), a, f(v)) = ([u]_\equiv, a, [v]_\equiv) \in
  E(\Sum{G})$. Conversely, suppose
  $([u]_\equiv, a, [v]_\equiv) \in E(\Sum{G})$ by virtue of
  $(u,a,v) \in E(G)$.
  Condition \eqref{eqn:indistinct:left} and
  \eqref{eqn:indistinct:right} of Lemma \ref{lem:indistinct}
  then imply
  that $(x,a,y) \in E(G)$ for every $x \in [u]_\equiv$ and
  $y \in [v]_\equiv$.
  
  \emph{(b) $G$ is not strongly homomorphic to any proper subgraph of
    \Sum{G}.}  By way of contradiction, suppose $G$ is strongly
  homomorphic to a proper subgraph $H$ of \Sum{G} via function
  $g : V(G) \rightarrow V(H)$.  Let $U$ contain exactly one vertex
  from each equivalence class induced by $\equiv$.  Let $f'$ be the
  restriction of $f$ to $U$. Function $f'$ is bijective. Since
  $V(H) \subset V(\Sum{G})$, the range of $g$ is a proper subset of
  the range of $f'$.  By the Pigeonhole Principle, there are two
  vertices $u,v \in U$, so that $g(u) = g(v)$ and $u \not\equiv v$,
  contradicting Lemma \ref{lem:sh_inverse}.
\end{proof}

\end{document}